\documentclass[aps,reprint,superscriptaddress,a4paper,nofootinbib,longbibliography]{revtex4-2}

\usepackage{amsthm}
\usepackage{amsmath}
\usepackage{amssymb}
\usepackage{bm}
\usepackage{braket}
\usepackage{float}
\usepackage{tcolorbox}
\usepackage{algorithm}
\usepackage{enumitem}
\usepackage{setspace}
\usepackage{hyperref}

\hypersetup{
    hidelinks,     
    colorlinks=true,   
    linkcolor=blue,    
    citecolor=blue,     
    urlcolor=magenta   
}
\setlist{noitemsep} 

\theoremstyle{plain}

\newtheorem{lemma}{Lemma}

\theoremstyle{definition}

\newcommand{\bars}{{\bar{\bm s}}}
\newcommand{\hP}{\widehat{P}}

\newcommand{\hPt}{\widehat{P^t}}
\newcommand{\hPts}{\widehat{P^t(\bm s)}}
\newcommand{\hatmu}{\hat\mu_{\bm x\cdot\bm s}}

\DeclareMathOperator{\Tr}{Tr}

\DeclareMathOperator{\var}{Var}

\begin{document} 

\title{Solving approximate hidden subgroup problems: 
quantum heuristics to detect \linebreak weak entanglement}

\author{Petar Simidzija}
\affiliation{Xanadu Quantum Technologies, Toronto, ON, M5G 2C8, Canada}

\author{Eugene Koskin}
\author{Elton Yechao Zhu}
\author{Michael Dascal}
\affiliation{Fidelity Center for Applied Technology, FMR LLC, Boston, MA, 02210, USA}

\author{Maria Schuld}
\affiliation{Xanadu Quantum Technologies, Toronto, ON, M5G 2C8, Canada}



\begin{abstract}
    How can we use a quantum computer to detect the entanglement structure of a quantum state? Bouland et al. (2024) recently provided an algorithm that, given multiple input copies of the state, finds the ``hidden cuts''---partitions into fully unentangled qubit registers. Their solution is based on turning cuts into a symmetry which can be detected with a Shor-type quantum algorithm for hidden subgroup problems, the \textit{hidden cut algorithm}. In this paper we derive heuristics that can find ``approximate symmetries'', or weakly entangled qubit registers, to unlock this powerful idea for a much broader range of problems. Our core contribution is a rigorous link between the output distribution of the hidden cut algorithm and the reward function that measures the quality of a cut. This implies that \textit{reducing} the number of state copies in the original hidden cut algorithm leads to measurement samples from which patterns of weak entanglement can be extracted. We believe that these insights are an important step in making quantum algorithms for hidden subgroup problems useful for applications beyond cryptography.
\end{abstract}

\maketitle

\section{Introduction}

The hidden subgroup problem (HSP) \cite{childs2010quantum} made quantum computing famous in the 1990s, when Shor linked an efficient quantum algorithm solving the problem to applications in cryptography. But finding less artificial use cases remains difficult up to this day \cite{babbush2025grand}, as it requires an oracle that loads a highly structured function ``hiding'' a subgroup of practical interest into a quantum state. Recently, Bouland, Giurgic{\u{a}}-Tiron and Wright \citep{bouland2025state} suggested a variation of the problem that could turn out to be an important step towards applications: instead of an oracle, the problem uses pairs of copies of an arbitrary $n$-qubit state $\ket{\psi}$ as input. A clever application of the standard HSP algorithm then discovers ``hidden cuts'' of the state, or a partition of the qubits into mutually non-entangled subsets. 

Of course, in realistic use cases we will hardly encounter perfect separability; a state of interest will more likely have a hierarchical pattern of partitions with increasingly weak entanglement, which we want to understand better, or perhaps variationally reinforce. For instance, we might be interested in characterizing or training the correlation structure of a generative quantum machine learning model where the qubits represent features \cite{liu2018differentiable, amin2018quantum, recio2025train}, or of a quantum simulation where the qubits model physical systems. But for weak entanglement, the hidden cut algorithm will produce the correct---and useless---answer: that there are no perfect cuts.

In this paper we take a deep look at the mechanism behind the quantum algorithm that solves the hidden subgroup problem to derive theoretically well-motivated heuristics that can deal with \textit{approximate} symmetries, and hence find weak entanglement structure. Our main insight is that the \textit{output distribution of the hidden cut algorithm is the Fourier transform of a reward function that measures bipartite entanglement}, which can be understood as the ``quality'' of a potential cut. The standard postprocessing strategy for HSP algorithms finds the maxima of this reward function. Although this is in principle quite difficult, the hidden cut algorithm uses a clever trick: increasing the number of input state pairs $t$ suppresses lower values of the reward function, which correspond to strongly entangled bipartitions. For large $t$, the reward function only has support in perfectly unentangled bipartitions or ``perfect cuts''. As we will see here, finding bipartitions into weakly entangled subsystems, or ``approximate cuts'', means finding \textit{secondary} maxima of this reward function.

The starting point for heuristics is to use enough input copy pairs to suppress most information in the reward function, but not the secondary maxima that signifies weak entanglement. The second ingredient of a heuristic is then to extract this information from the samples of the quantum computer by changing the postprocessing strategy of the original hidden cut algorithm, which is the standard postprocessing used in Shor's and other HSP algorithms. We derive two such strategies: The first one stops postprocessing early in order to find a solution other than the trivial one (i.e., that there are no perfect cuts), and the second one constructs an estimator of the quality of a cut which can be efficiently evaluated---and optimized---on a classical computer.

Although the use of these heuristics for applications will have to be established on a case-by-case basis, they have several attractive properties: The circuit only requires $2n$ Hadamard and $n$ controlled SWAP gates, as well as few input state copies. The number of input states constitutes a powerful hyperparameter that polynomially ``amplifies'' what we can extract about the entanglement structure (similar to cost amplification in Decoded Quantum Interferometry \cite{jordan2025optimization}). Lastly, the heuristics are theoretically well motivated, since we provide an analytic expression that relates the measurement distribution of the quantum algorithm to the (un-)entanglement structure of the input state--knowledge that can also inform approaches to variationally enforcing this structure. A clear limitation for application is the need to prepare simultaneous copies of the input state at all: while qubits are still a scarce resource, this is a serious demand. 

The next section provides a short technical summary of the main results. We then introduce and thoroughly analyze the original setting of the hidden cut algorithm (Section~\ref{sec:background}) and derive the link between output distribution and the reward function (Section~\ref{sec:result}). Section~\ref{sec:heuristics} describes the heuristics in more detail. We motivate the relevance of entanglement detection in the conclusion.

\section{Technical summary}

Our main technical result relates the output distribution of the hidden cut algorithm to what we called above the ``reward function'' of hidden cut finding. When run with a single input pair, $t=1$, this reward function is the purity 
\begin{align}
    P(\bm s) \equiv \Tr(\rho_{\bm s}^2)
\end{align} of the reduced state of a subsystem $\bm s$. We will define a subsystem as  a bitstring $\bm s\in \{0,1\}^n$ by marking included qubits with 1; for example, $\bm s=11000$ defines the subsystem of qubits $\{1, 2\}$ and $P(11000)$ is the purity of the reduced state $\rho_{12}$. Two complementary bitstrings, such as $\bm s=11000$ and $\bar{\bm s} = 00111$, then define the same bipartition (i.e., $\{1, 2\} \{3, 4, 5\}$), and have the same purity. 

The purity can be seen as a measure for the weakness of the entanglement between two qubit registers, or the ``strength of a cut''. A \textit{perfect cut} is then defined as a complementary bitstring pair $\bm s, \bar{\bm s}$ that has maximum purity: $P(\bm s) = P(\bar{\bm s}) = 1$. The purity of the all-zeros and all-ones bitstring (signifying no cuts) is always $1$, which we refer to as the \textit{trivial cut}. An \textit{approximate cut} is represented by a complementary bitstring pair where $P(\bm s) = P(\bar{\bm s}) \approx 1$. 

Our main technical contribution is to show that the hidden cut algorithm with $t$ pairs of input copies finds maxima of the purity function by sampling from an output distribution that is the \textit{Fourier transform} of its $t$'th power,
\begin{align}\label{eq:main_result}
    p_t(\bm x) = \frac{1}{2^n} \sum_s P^t(\bm s) (-1)^{\bm x\cdot \bm s},
\end{align}
where $\bm x \bm s$ is the dot product modulo $2$. Taking the purity to the $t$'th power has the effect of suppressing small purities (i.e., bipartitions of high entanglement), increasingly erasing their contribution to the output distribution of the quantum algorithm. Small values are suppressed much faster than large ones, but in the limit of many input state pairs $t$, only information on perfect cuts survives. The number $t$ of input state copies therefore acts as a \textit{hyperparameter that accentuates the reward function} (see Figure~\ref{fig:regimes}).

\begin{figure}
    \centering
    \includegraphics[width=\linewidth]{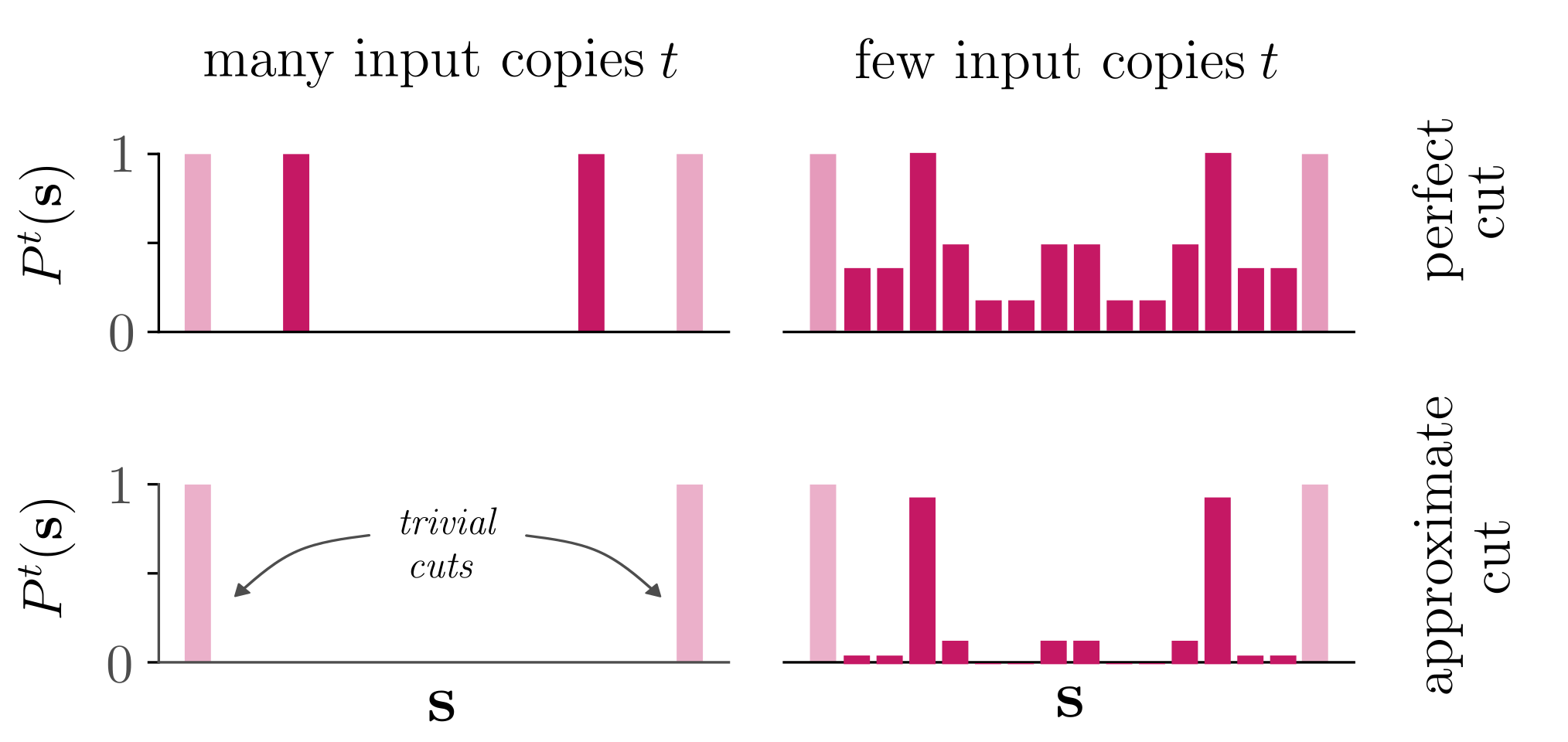}
    \caption{Finding weakly entangled bipartitions of qubits means finding subsystems $\bm s$ whose reduced state has a high purity $P(\bm s)$. We show here that the output distribution of the hidden cut algorithm run with $2t$ copies \cite{bouland2025state} is the Fourier transform of $P^t(\bm s)$. Using many input copies attenuates all purities smaller than $1$ and guarantees that postprocessing the measurement samples reveals $\bm s$ with $P(\bm s)=1$ (top left). However, it also suppresses information on ``approximate cuts'' or $P(\bm s)<1$ (bottom left). To build heuristics, we need to use a small number of input copies, which keeps information on secondary maxima, but still suppresses low values (bottom right). This example uses $t=300$ (left) and $t=5$ (right) for a product of two Haar random states, $\ket{\psi} = \ket{\psi_{01}}\ket{\psi_{23}}$, which we mixed in $0.1$ parts with a Haar random state to create a weak entanglement bipartition at $\bm s = 0011$, $\bar{\bm s}=1100$.  The code to reproduce all figures can be found here: \url{https://github.com/XanaduAI/approximate_hidden_subgroups}.}
    \label{fig:regimes}
\end{figure}

In the standard postprocessing procedure of the HSP algorithm, every sample from the output distribution $p_t(\bm x)$ eliminates all solution candidates $\bm s$ with $\bm x\cdot \bm s = 1$ (instead of $\bm x\cdot \bm s = 0$). As such, every previously unseen sample eliminates half of all solution candidates, which means that a few measurement shots suffice to reveal the set of correct solutions. When there are no perfect cuts, this elimination strategy will discard all but the trivial solution. Our core insight for building heuristics is that, for small $t$, Eq.~(\ref{eq:main_result}) implies that the probability of a sample is higher if, on average, it eliminates subsystems $\bm s, \bar{\bm s}$ with lower ``reward'' $P^t(\bm s)=P^t(\bar{\bm s})$. Hence, \textit{the samples we see likely eliminate poorer solution candidates during post-processing}.

We derive two concrete heuristics, both running with a carefully tuned, but small, number $t$ of input state pairs. The first heuristic simply stops eliminating solution candidates $\bm s$ before only the trivial solution remains. As the order of samples influences the result, this is repeated several times, revealing bipartitions that have a high likelihood of being weakly entangled. The second heuristic defines a function of how close $\bm x\cdot \bm s$ is to zero, averaged over all measurement samples $\bm x$. This function $\widehat{P^t}(\bm s)$ turns out to be an estimator of the reward function $P^t(\bm s)$, and allows us to probe any desired ``cut quality'' on a classical computer, which makes the entanglement structure of a quantum state accessible for classical optimization routines. 

\section{The hidden cut algorithm}
\label{sec:background}

The backbone of our approach to estimating weak entanglement structures is the hidden cut circuit.
This circuit was originally used in \cite{bouland2025state} to solve the \textit{hidden cut problem},
which can be viewed as a special case of the more general \textit{state hidden subgroup problem},
which itself is a modification of the traditional \textit{hidden subgroup problem} (i.e., \cite{childs2010quantum}).
Here we provide a brief recap of this background material, and refer to our PennyLane demo for a more didactic introduction \cite{PetarSimidzija2025}.
Readers already familiar with the material can skip this section.

\subsection{Hidden subgroup problem}

The hidden subgroup problem (HSP) is the following: 
given oracle access to a function $f(g)$ over a group $G$, 
such that a subgroup $H$ of $G$ is a symmetry of $f$, find $H$.
If all we are given is oracle access to $f$,
then in general we would need $\mathcal O(|G|)$ queries to determine $H$.
However, if instead of oracle access we assume the ability to prepare states $\ket{f(g)}$
with the property $\braket{f(g)|f(g')}=0$ if $f(g)\neq f(g')$,
then it is possible to solve the HSP for abelian $G$ in $\mathcal O(\log |G|)$ time using a quantum computer.

The most well-known problem that can be formulated as an HSP, and thus solved efficiently on a quantum computer,
is integer factorization~\cite{shor1994algorithms, shor1999polynomial}.
Unfortunately, besides factoring and a few other cryptography-motivated problems~\cite{hallgren2007polynomial, biasse2016efficient, Childs_2010, horan2018hidden},
it has been difficult to find other applications that benefit from a quantum computer's ability to efficiently solve abelian HSPs.
Perhaps one reason for this is that the above quantum formulation of the HSP is somewhat unnatural, 
since it assumes that there is a \textit{classical} function that hides the symmetry,
but requires us to be able to prepare \textit{quantum} states that encode this function.

\subsection{State hidden subgroup problem}

An arguably more natural quantum formulation of the HSP was recently proposed in~\cite{bouland2025state},
and called the \textit{state hidden subgroup problem} (StateHSP).
Unlike the the traditional quantum formulation of the HSP, StateHSP makes no reference to a classical function $f$, 
and instead simply assumes access to copies of a quantum state $\ket\Psi$ that has an unknown symmetry group $H$.
More precisely, it is assumed that $\ket\Psi$ is an element of a Hilbert space on which a unitary representation $U(g)$ of $G$ acts, 
such that $U(g)\ket\Psi$ and $\ket\Psi$ are equal if $g\in H$ and orthogonal otherwise.

\begin{figure}[t!]
    \centering
    \includegraphics[width=0.4\textwidth]{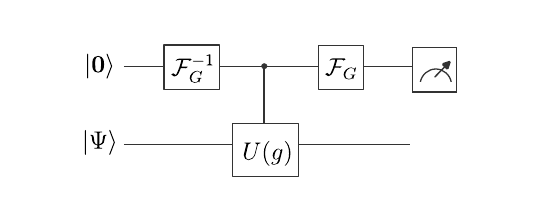}
    \caption{StateHSP circuit for group $G$.}
    \label{fig:StateHSP_circuit}
\end{figure}

The algorithm solving the StateHSP works by applying a quantum circuit to $\ket\Psi$ so that a few measurement samples are sufficient to find the hidden subgroup $H$.
The circuit, shown in Figure~\ref{fig:StateHSP_circuit}, has two registers.
The first, or \textit{group}, register
is a $|G|$-dimensional Hilbert space which we can think of as either encoding the group $G$
or the irreducible representations (irreps) of $|G|$.\footnote{
A standard result for finite groups is that $\sum_\rho(\text{dim }\rho)^2=|G|$,
where the sum is over irreps $\rho$, 
and hence the number of matrix elements in all the irreps is equal to the size of the group.
Therefore a vector space encoding $G$ is the same size, and thus isomorphic to, a vector space encoding the irreps of $G$.
The $G$-Fourier transform and its inverse map between the canonical bases of these vector spaces.
}
The group register is initialized to the state $\ket{\bm 0}$, 
corresponding to the trivial irrep.
The second, or \textit{state}, register is initialized to the state $\ket\Psi$.
The circuit first applies the inverse $G$-Fourier transform $\mathcal F_G^{-1}$ to the group register, 
after which this register is in the uniform superposition $\sqrt{|G|}^{-1}\sum_{g\in G}\ket g$.
Next, the unitary $U(g)$ is applied to the state register, 
with the value of $g$ conditioned on the state of the group register.
Finally, the $G$-Fourier transform $\mathcal F_G$ is applied to the group register, and this register is measured.

When $G$ is an abelian group, 
the measurements from the quantum circuit can always be efficiently postprocessed to determine the hidden subgroup $H$.
The postprocessing is the same as for the HSP (see e.g.~\cite{Childs_2010}).
The key fact is that for abelian groups $G$, measurements of the circuit are samples from the \textit{annihilator} $H^{\perp}$, 
defined as the set of homomorphisms from $G$ to $\mathbb C$ that map all elements in $H$ to 1.
Given the matrix $M$ constructed from samples of $H^\perp$, it is straightforward to determine $H$ as the nullspace of $M$.

\subsection{Hidden cut problem}\label{sec:background_hc}

\begin{figure}[t!]
    \centering
    \includegraphics[width=0.35\textwidth]{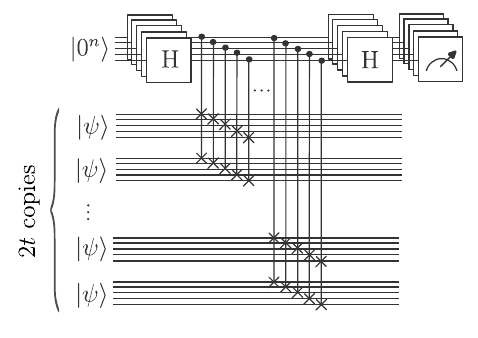}
    \caption{General $t$-copy hidden cut circuit.}
    \label{fig:hc_circuit_t}
\end{figure}

The hidden cut problem is the following: given an $n$-qubit state $\ket\psi$,
determine the partition of the qubits for which there is no entanglement between different components.
An efficient solution to this problem follows from recognizing that it can be viewed as an instance of StateHSP over the group $G=\mathbb Z_2^n$ \cite{bouland2025state}.

To see this reduction, consider the case where $\ket\psi=\ket{\psi_{\bm s^*}}\ket{\psi_{\bars^*}}$ consists of 
two disentangled subsystems $\bm s^*$ and $\bars^*$ (the case with more than two components is similar).
Here we are using our notation from earlier of denoting subsystems of $n$ qubits with $n$-bit strings.
Note that $n$-bit strings can also be used to canonically represent elements of $G=\mathbb Z_2^n$.
In this way $\bm s^*$ and $\bars^*$ can be viewed as generators of the $\mathbb Z_2\times \mathbb Z_2$ subgroup 
$H\equiv\{\bm 0,\bm s^*, \bars^*, \bm 1\}$ of $G$,
where $\bm 0$ is the all 0 bitstring (group identity) and $\bm 1$ is the all 1 bitstring.
Note that if we somehow found $H$, we would also find $\bm s^*$ and $\bars^*$, and would thereby solve the hidden cut problem. 

Indeed, it is possible to use the StateHSP algorithm to find $H$.
Recall that to do so we need to define a state $\ket\Psi$ and a representation $U(\bm s)$ of $G$
such that $U(\bm s)\ket\Psi$ and $\ket\Psi$ are equal if $\bm s\in H$, and orthogonal otherwise.
Let us define $\ket\Psi=(\ket\psi\ket\psi)^{\otimes t}\in (\mathcal H\otimes\mathcal H)^{\otimes t}$, 
where $\mathcal H$ is the Hilbert space of $\ket\psi$ (we will specify the exponent $t$ shortly).
Furthermore, define the unitary $U(\bm s)\equiv \text{SWAP}_{\bm s}^{\otimes t}$ on 
$(\mathcal H\otimes\mathcal H)^{\otimes t}$,
where we recall that $\text{SWAP}_{\bm s}$ acts on $\mathcal H\otimes \mathcal H$ by swapping the $\bm s$ qubits in the first $\mathcal H$ with the $\bm s$ qubits in the second $\mathcal H$.

Now consider $U(\bm s)\ket\Psi = \left(\text{SWAP}_{\bm s}\ket\psi\ket\psi\right)^t$.
We can consider each factor $\text{SWAP}_{\bm s}\ket\psi\ket\psi$ separately.
If $\bm s\in H$, there are four options:
\begin{itemize}
    \item $\bm s=\bm 0$: nothing happens, 
    \item $\bm s=\bm s^*$: the $\bm s^*$ components of the two $\ket\psi$s get swapped,
    \item $\bm s=\bars^*$: the $\bars^*$ components of the two $\ket\psi$s get swapped, and
    \item $\bm s=\bm 1$: the entire first $\ket\psi$ gets swapped with the entire second $\ket\psi$.
\end{itemize}
In all cases $U(\bm s)\ket\Psi$ is equal to $\ket\Psi$, as desired for $\bm s\in H$.
Meanwhile, if $\bm s\notin H$, we need $U(\bm s)\ket\Psi$ to be orthogonal to $\ket\Psi$. 
This is true only if $\left(\text{SWAP}_{\bm s}\ket\psi\ket\psi\right)^t$ is orthogonal to $(\ket\psi\ket\psi)^t$, 
which is not true in general.
But by taking $t$ to be large, any overlap between the two states will decay exponentially to zero,
and so $U(\bm s)\ket\Psi$ is effectively orthogonal to $\ket\Psi$ for large $t$.
Therefore the state $\ket\Psi$ and representation $U(g)$ satisfy the StateHSP requirements,
and thus the StateHSP circuit in Figure~\ref{fig:StateHSP_circuit}, 
which in this case reduces to the hidden cut circuit (Figure~\ref{fig:hc_circuit_t}),
can be used to solve the hidden cut problem.

\section{Main result: Relating subsystem purities to measurement samples}\label{sec:result}

We have seen how the hidden cut circuit can be used to solve the hidden cut problem.
We will now prove the main result used to derive heuristics in the next section, namely that the output distribution $p_{t}(\bm x)$ of the hidden cut circuit is 
the inverse $\mathbb Z_2^n$-Fourier transform of the subsystem purity function $P(\bm s)$, raised to the $t$'th power. To prove this statement, we will first derive another result, namely that $p_{t}$ is the $t$-fold convolution of $p_{1}$, the output distribution of the hidden cut circuit with $t=1$. 

Note that the results derived in this section can be generalized to StateHSP circuits over arbitrary abelian groups (which has also been discussed in \cite{hinsche2025abelian}),
which we summarise in Appendix~\ref{app:abelian_StateHSP}.

\subsection{Relating the $t$-copy measurement distribution to $t=1$ via convolution}

Let us define the measurement distribution of the $t$-copy hidden cut circuit as $p_t(\bm x)$,
which is a distribution over $n$-bit strings $\bm x$, corresponding to measurements of the circuit's group register.
We begin by proving that for arbitrary $t$, $p_t(\bm x)$ is equal to the $t$-fold $\mathbb Z_2^n$-convolution of $p_1(\bm x)$.

To prove this statement, consider the hidden cut circuit with parameter $t$, acting on $t$ copies of $\ket{\Psi_1}\equiv \ket\psi\ket\psi$ (Figure~\ref{fig:hc_circuit_t}).
The final pre-measurement state is
\begin{align}
    \frac{1}{\sqrt{2^n}}\sum_{\bm s\in \mathbb Z_2^n}
    \sum_{\bm x\in \mathbb Z_2^n}
    H_{\bm x\bm s}\ket{\bm x}
    \left(U(\bm s)\ket{\Psi_1}\right)^t,
\end{align}
where $U(\bm s)\equiv \text{SWAP}_{\bm s}$ is the unitary swapping subsystem $\bm s$ between the two copies of $\ket\psi$ in each of the $t$ pairs, 
and $H_{\bm x\bm s} \equiv \frac{1}{\sqrt{2^n}}(-1)^{\bm x\cdot \bm s}$ are the components of the $n$-th order Hadamard matrix,
with $\bm x\cdot \bm s$ the dot-product modulo 2.
Measuring the group register results in outcome $\bm x$ with probability
\begin{align}
    p_t(\bm x) = \frac{1}{2^n}\sum_{\bm s'} \sum_{\bm s''} H_{\bm x\bm s'} H_{\bm x\bm s''}
    \left(\bra{\Psi_1} U^\dagger(\bm s')U(\bm s'')\ket{\Psi_1}\right)^t.
\end{align}
We can simplify this expression by noting that
since $U$ is a unitary representation of $G=\mathbb Z_2^n$ we have that $U^\dagger(\bm s)=U^{-1}(\bm s)=U(\bm s^{-1})=U(\bm s)$.
Therefore 
\begin{align}
    p_t(\bm x) = 
    \frac{1}{2^n}\sum_{\bm s'} \sum_{\bm s''} 
    \frac{1}{\sqrt{2^n}}H_{\bm x(\bm s' \oplus \bm s'')} 
    \bra{\Psi_1} U(\bm s' \oplus \bm s'')\ket{\Psi_1},
\end{align}
where $\oplus$ denotes bitwise summation modulo 2.
Letting $\bm s\equiv \bm s' \oplus \bm s''$, the double sum simplifies to a single sum,
\begin{align}\label{eq:p_t_from_Psi_v0}
    p_t(\bm x) &= \frac{1}{\sqrt{2^n}} \sum_{\bm s} H_{\bm x\bm s}
    \left(\bra{\Psi_1} U(\bm s)\ket{\Psi_1}\right)^t.
\end{align}

Now let us define the (classical) Fourier and inverse Fourier transforms for a function $f$ over $G=\mathbb Z_2^n$ as
\begin{align}
    \mathcal F[f](\bm s) 
    & \equiv \sum_{\bm x} (-1)^{\bm x\cdot \bm s} f(\bm x) \label{eq:FT}\\
    & = \sqrt{2^n}\sum_{\bm x} H_{\bm s\bm x} f(\bm x),\\
    \mathcal F^{-1}[\hat f](\bm x) 
    & \equiv \frac{1}{2^n}\sum_{\bm s} (-1)^{\bm x\cdot \bm s} \hat f(\bm s) \label{eq:FT_inv}\\
    & = \frac{1}{\sqrt{2^n}}\sum_{\bm s} H_{\bm x\bm s} \hat f(\bm s).
\end{align}
Additionally, we define the convolution of two functions $f$ and $g$ over bitstrings as
\begin{align}
    (f\ast g)(\bm x)\equiv \sum_{\bm x'} f(\bm x')g(\bm x\oplus\bm x'),
\end{align}
With our prefactor conventions for $\mathcal F$ and $\mathcal F^{-1}$, the convolution theorem states that
\begin{align}
    f\ast g = \mathcal F^{-1}\big[\mathcal F[f]\mathcal F[g]\big].
\end{align}

With these definitions we can write Eq.~\eqref{eq:p_t_from_Psi_v0} for $p_t$ as
\begin{align}\label{eq:p_t_from_Psi}
    p_t =
    \mathcal F^{-1}[\left(\bra{\Psi_1} U(\cdot) \ket{\Psi_1}\right)^t].
\end{align}
Setting $t=1$ we find
\begin{align}\label{eq:p_1_FT}
    p_1 = \mathcal F^{-1}[\bra{\Psi_1} U(\cdot)\ket{\Psi_1}],
\end{align}
which we can invert to get
\begin{align}
    \bra{\Psi_1} U(\bm s) \ket{\Psi_1} =
    \mathcal F[p_1(\cdot)](\bm s).
\end{align}
Substituting this into the Eq.~\eqref{eq:p_t_from_Psi} gives $p_t$ in terms of $p_1$:
\begin{align}\label{eq:pt=Finv[F[p1^t]]}
    p_t = \mathcal F^{-1}
    \left[\mathcal F[p_1]^t\right],
\end{align}
or, by the convolution theorem,
\begin{align}\label{eq:pt=p1*t}
    p_t = p_1^{\ast t}
    \equiv 
    \underbrace{p_1 \ast \cdots \ast p_1}_{\text{$t$ times}}.
\end{align}
In words, $p_t$ is the $t$-fold $\mathbb Z_2^n$-convolution of $p_1$.

We will use this result in the next section. Before moving on, note that one useful implication of this result
is that it allows for the classical simulation of the hidden cut circuit for any $t$, 
as long as it is possible to compute the distribution for $t=1$.
Further, we can see that $p_t=p_1^{\ast t}$ implies that the support of $p_t$ is equal to the support of $p_1$, 
and the solution to the hidden cut problem is encoded precisely in this support---a fact we will investigate further below.

\subsection{Relating the measurement distribution to subsystem purities via Fourier transform}\label{sec:fourier_hidden}

We will now show that $p_{t}(\bm x)$ is the inverse $\mathbb Z_2^n$-Fourier transform of $P^t(\bm s)$, the $t$'th power of the subsystem purity function that can be understood as the reward function of finding strong cuts\footnote{The most well-known measure of entanglement is the von Neumann entropy $S(\bm s)\equiv-\Tr(\rho_{\bm s}\log\rho_{\bm s})$,
but because of the logarithm this can be difficult to compute.
The purity is a common alternative, which can be understood as a linearization of $1-S(\bm s)$. }.

To see this, recall Eq.~\eqref{eq:p_1_FT}:
\begin{align}\label{eq:p_1_FT_v2}
    p_1(\bm x) = \mathcal F^{-1}[\bra{\Psi_1} U(\cdot)\ket{\Psi_1}](\bm x).
\end{align}
This tells us that $p_1(\bm x)$ is the inverse Fourier transform of $\bra{\Psi_1} U(\bm s)\ket{\Psi_1}$, 
where $\ket{\Psi_1}\equiv \ket\psi\ket\psi$ and $U(\bm s)\equiv \text{SWAP}_{\bm s}$.
Defining $\rho_{\bm s} \equiv \Tr_{\bar{\bm s}}\ket\psi\bra\psi$,
the reduced state of $\ket{\Psi_1}$ on the qubits in $\bm s$ from the two copies of $\ket\psi$ is $\rho_{\bm s}\otimes\rho_{\bm s}$.
Therefore
\begin{align}
    \bra{\Psi_1} U(\bm s)\ket{\Psi_1} = 
    \Tr\left[
    \text{SWAP}
    \left(\rho_s\otimes\rho_s\right)
    \right],
\end{align}
where, because we have already reduced to the $\bm s$ qubits, 
$\text{SWAP}$ is now the regular swap operator, which swaps all qubits in the first tensor product factor with all qubits in the second tensor product factor.
Writing out the trace in an orthonormal basis, we find
\begin{align}
    \bra{\Psi_1} U(\bm s)\ket{\Psi_1}
    &= \sum_a\sum_b\bra a\bra b(\rho_{\bm s}\otimes\rho_{\bm s})\text{SWAP}\ket a\ket b\\
    &= \sum_a\sum_b\bra a\rho_{\bm s}\ket b\bra b\rho_{\bm s}\ket a\\
    &= \sum_a \bra a\rho_{\bm s}^2\ket a\\
    &= P(\bm s).\label{eq:swap=purity}
\end{align}
Therefore we conclude that 
\begin{align}
    p_1 =\mathcal F^{-1}[P],
\end{align}
or, in words, $p_1$ is the inverse $\mathbb Z_2^n$-Fourier transform of the subsystem purity function $P$.
Using our result \eqref{eq:pt=Finv[F[p1^t]]} that $p_t = \mathcal F^{-1}\left[\mathcal F[p_1]^t\right]$, we get
\begin{align}\label{eq:pt_2}
    p_t = \mathcal F^{-1}
    \left[P^t\right],
\end{align}
or, in words, $p_t$ is the inverse $\mathbb Z_2^n$-Fourier transform of $P^t$.

\subsection{Understanding the result better}\label{sec:structure}

Let us make the structure of $P^t(\bm s)$, the subsystem purity or ``reward'' function raised to the power of $t$, 
and $p_t(\bm x)$, the output distribution of the hidden cut circuit, more tangible. This will also help to illustrate the effect of the number of copies $t$, and to justify the heuristics suggested in the following sections.

As we have seen, these two functions over $n$-bit strings are related by a $\mathbb Z_2^n$-Fourier transform.
While the exact forms of the functions will depend on the state $\ket\psi$, 
we can nevertheless make some useful remarks that hold for general classes of states.
In particular we will find that the entanglement structure of $\ket\psi$ determines 
important global features of $P^t$ and $p_t$, namely their support and their uniformity (i.e.~entropy).
To that end, we will consider as examples two 4-qubit states $\ket\psi$ with different entanglement structures: 
\begin{enumerate}
    \item \textbf{Non-separable $\ket\psi = \ket{\psi_{0123}}$}: Haar-random 4-qubit state,
    \item \textbf{Separable $\ket\psi = \ket{\psi_{01}}\ket{\psi_{23}}$}: product of two Haar-random $2$-qubit states.
\end{enumerate}
Plots of $P^t(\bm s)$ and $p_t(\bm x)$, for $t\in\{1,3,100\}$, are shown in Figure~\ref{fig:psi_ent} for an instance of the non-separable state 
and Figure~\ref{fig:psi_sep} for an instance of the separable state.

\begin{figure}[t!]
    \centering
    \includegraphics[width=0.5\textwidth]{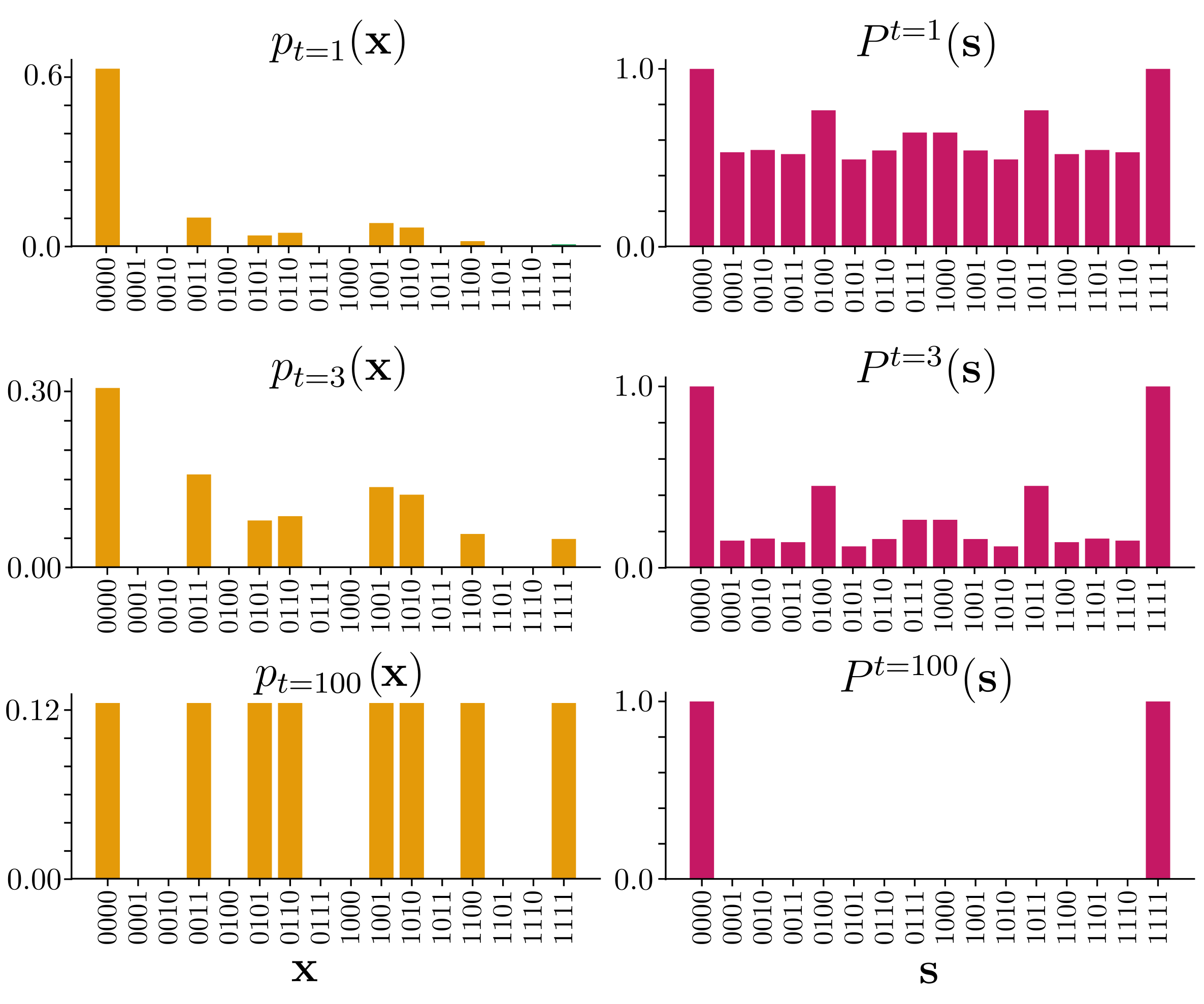}
    \caption{Non-separable state $\ket\psi=\ket{\psi_{0123}}$: Subsystem purities $P^t(\bm s)$ and hidden cut circuit output distribution $p_t(\bm x)$ for a Haar-random 4-qubit state, which are related by a Fourier transform. The output distribution of the hidden cut algorithm has support over bitstrings that are orthogonal to bitstrings for which $P(\bm s)=1$. Here, the support is over all bitstrings of even parity, which are only orthogonal to $\bm s = 0000$ and $\bm s = 1111$. Increasing $t$ in the algorithm by adding state copies accentuates these features: $p_t$ becomes a uniform distribution over the support of $p_1$, and $P^t$ only keeps support for subsystems of purity $1$. }
    \label{fig:psi_ent}
\end{figure}

First consider $t=1$.
For $\ket\psi=\ket{\psi_{0123}}$ the subsystem purity function $P(\bm s)$ is plotted in the top-left panel of Figure~\ref{fig:psi_ent}.
Because this state is non-separable we see that the purities of all non-trivial subsystems are strictly less than 1,
i.e.~the only pure subsystems are the empty subsystem $\bm s=0000$ (which is pure by definition) and the full system $\bm s = 1111$ (which is pure because $\ket\psi$ is pure).
Meanwhile for $\ket\psi=\ket{\psi_{01}}\ket{\psi_{23}}$, 
whose purity function is plotted in the top-right panel of Figure~\ref{fig:psi_sep},
we see that, in addition to these two trivial subsystems,
the non-trivial subsystems $\bm s = 0011$ and $\bm s = 1100$ also have purity equal to 1.

For both states the subsystems of purity 1 form a subgroup $H$ of the full group of $n$-bit strings $G=\mathbb Z_2^n$.
For the non-separable state this group is $H=\mathbb Z_2=\{0000,1111\}$,
while for the separable state it is $H=\mathbb Z_2\times \mathbb Z_2=\{0000,0011,1100,1111\}$.
In either case the generators of $H$ directly encode the separable components of the state, 
i.e.~give us the solution to the hidden cut problem.
For $\ket\psi=\ket{\psi_{0123}}$ the single generator is $1111$, corresponding to the fact that no cut separates the state,
while for $\ket\psi=\ket{\psi_{01}}\ket{\psi_{23}}$ the generators are $0011$ and $1100$, corresponding to the pure subsystems $\ket{\psi_{01}}$ and $\ket{\psi_{23}}$.
Furthermore, notice that in both cases $H$ is the  \textit{symmetry subgroup} of the corresponding purity function,
that is $P(\bm s \oplus \bm h)=P(\bm s)$ for any $\bm s\in G$ and $\bm h\in H$.
In other words, the subsystem purity function has a special property that general functions over bitstrings do not have:
its symmetry subgroup is formed by the bitstrings with function value equal to 1.

\begin{figure}[t!]
    \centering
    \includegraphics[width=0.5\textwidth]{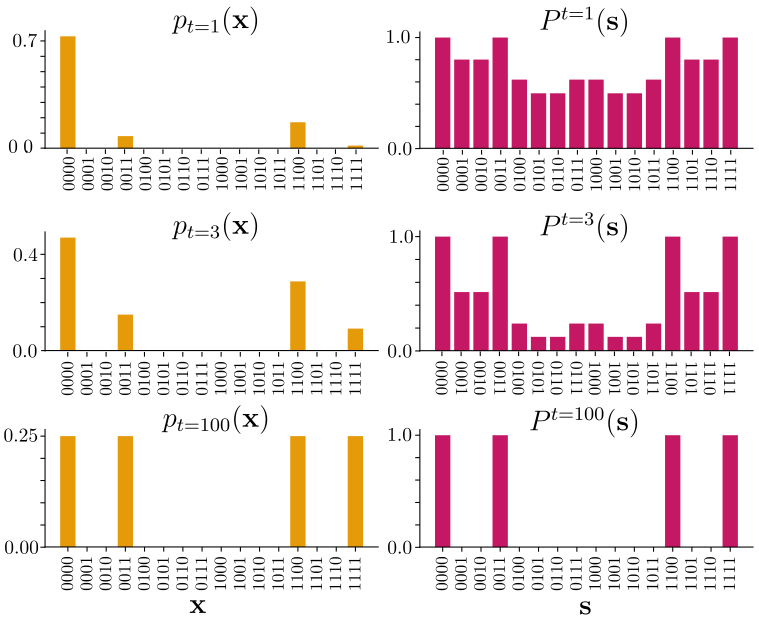}
    \caption{Separable state $\ket\psi=\ket{\psi_{01}}\ket{\psi_{23}}$: Subsystem purities $P^t(\bm s)$ and hidden cut circuit output distribution $p_t(\bm x)$ for a product of two Haar-random 2-qubit state. The output distribution of the hidden cut algorithm for this state has support on the bitstrings $0000$ and $1111$ as well as $0011$ and $1100$. The same bitstrings are orthogonal to these, and represent the subsystems with purity $1$, hence revealing the ``hidden cut''. }
    \label{fig:psi_sep}
\end{figure}

Given that the cut structure of the state is encoded in the symmetry group $H$ of $P(\bm s)$, 
it might be tempting to directly look for this structure by computing the purities of various subsystems and
thereby deducing $H$, for example with a certain kind of swap circuit.
Unfortunately this approach would be very inefficient.
First of all, the purities would have to be estimated with high enough precision to accurately deduce $H$,
which would potentially require many shots of the swap test circuit.
But more significantly, to determine the symmetry group of $P$ in this way for a state with $k$ separable components,
we would need to estimate $\mathcal O(2^{n-k})$ subsystem purities,
which is only feasible if there are many non-entangled components, i.e.~$k\approx n$.

Fortunately there is a much more efficient way to determine the symmetry subgroup of $P$.
The general insight is that if a function has a symmetry, then its Fourier modes must also have this symmetry.
Recalling that the inverse $\mathbb Z_2^n$-Fourier transform of $P$ is $p_1$,
and denoting the Fourier modes as $\chi_{\bm x}(\bm s)$,
the Fourier modes present in $P$ are those $\chi_{\bm x}$ for which $p_1(\bm x)>0$.
Now, the condition that a Fourier mode is symmetric under $H$ is that 
$\chi_{\bm x}(\bm s\oplus\bm h)=\chi_{\bm x}(\bm s)$ for all $\bm s\in G$ and $\bm h\in H$.
Because $\chi_{\bm x}$ is a homomorphism of $G$ into $\mathbb C$, 
this is equivalent to the statement that $\chi_{\bm x}(\bm h)=1$ for all $h\in H$,
i.e.~$\chi_{\bm x}$ is in the annihilator $H^\perp$.
For the group $\mathbb Z_2^n$ we have $\chi_{\bm x}(\bm h) = (-1)^{\bm x\cdot\bm h}$,
so $\chi_{\bm x}(\bm h)=1$ is equivalent to $\bm x\cdot \bm h = 0$.
Therefore the elements of $H$ are those bitstrings which are orthogonal to the support of $p_1(\bm x)$. 
Thus we can determine $H$ by sampling from $p_1(\bm x)$, organizing the measurement samples as rows in a \textit{measurement matrix} $M$,
and computing the nullspace of $M$.
Since every sample reduces the candidate for the nullspace by half, we are guaranteed that only a logarithmic number of samples (and in practice, very few) are sufficient to compute this nullspace, and hence solve the hidden cut problem.

For example, we see in the top-left plot of Figure~\ref{fig:psi_ent} that the support of $p_1$ for the non-separable state $\ket\psi=\ket{\psi_{0123}}$ consists of all even-parity bitstrings $\bm x$ (those with an even number of 1s).
There are only two $\bm s$ which are orthogonal to all even-parity bitstrings: $\bm s=0000$ and $\bm s=1111$.
Indeed we have already seen that these form the hidden subgroup $H$ for the non-separable state.
Now consider the separable state $\ket\psi=\ket{\psi_{01}}\ket{\psi_{23}}$, for which the distribution $p_1$ is shown in the top-left plot of Figure~\ref{fig:psi_sep}.
We see that the support of $p_1$ is sparser than in the case of the non-separable state: 
it is given by $H^\perp=\{0000, 0011,1100,1111\}$.
In this case $H=\{0000, 0011,1100,1111\}$ (same as $H^\perp$),
which is indeed the subgroup we expect for $\ket\psi=\ket{\psi_{01}}\ket{\psi_{23}}$.

We can summarize our conclusions for the 1-copy hidden cut circuit by saying that the following are all equivalent:
\begin{enumerate}
    \item \textbf{Hidden subgroup $H$}: $\bm s\in G$ for which $\bm s$ is an unentangled subset of $\ket\psi$,
    \item \textbf{Pure subsets}: $\bm s\in G$ for which $P(\bm s)=1$,
    \item \textbf{Symmetries of $P$}: $\bm s\in G$ for which $P(\bm s'\oplus \bm s)=P(\bm s')$ for all $\bm s'\in G$, and
    \item \textbf{Nullspace of $p_1$}: $\bm s\in G$ for which $\bm s\cdot \bm x=0$ for all $\bm x$ with $p_1(\bm x)>0$.
\end{enumerate}

Having understood the case $t=1$, it is straightforward to understand the structure of $P^t$ and $p_t$ for general $t$.
First, because $0\le P(\bm s)\le 1$, when we raise $P$ to higher powers, the values of $P(\bm s)$ which are less than 1 will decay to 0,
while the values which are equal to 1 will remain the same.
Since we already established that the $\bm s$ for which $P(\bm s)=1$ is precisely the hidden subgroup $H$, 
this means that in the limit $t\rightarrow\infty$ the spectrum of $P^t(\bm s)$ will consist solely of $H$.
We can demonstrate this in the left-side plots of Figs.~\ref{fig:psi_ent} and \ref{fig:psi_sep},
where we increase $t$ from 1 to 3 to 100.
At $t=100$ the support of $P^t$ effectively consists only of the hidden subgroup $H$ of the respective state $\ket\psi$.

Now consider $p_t$ as we increase $t$.
From the right-side plots of Figs.~\ref{fig:psi_ent} and \ref{fig:psi_sep} we notice two things.
First, the support of $p_t$ remains the same as we increase $t$.
And second, the values of the non-zero $p_t(\bm x)$ \textit{uniformize}, that is they approach the same value in the limit $t\rightarrow\infty$.
This value is of course determined by the normalization $\sum_{\bm x}p_t(\bm x)=1$.
Recalling that the support of $p_1$ is $H^\perp$, we conclude that in the limit of large $t$ the function $p_t$ is the uniform distribution on $H^\perp$.

We can explain this uniformization analytically in two ways.
The simplest way is to note that because $p_t$ is the inverse Fourier transform of $P^t$, 
the uniformization of $P^t$ to $H$ in the limit of large $t$ corresponds to the uniformization of $p_t$ to $H^\perp$. 
Alternatively, we can use the fact that $p_t$ is the $t$-fold self-convolution of the probability distribution $p_1$,
and make use of general results that show that repeated self-convolution of a probability distribution leads to a uniform distribution (over the supported values).\footnote{
See the textbook by Diaconis, Ref.~\cite{diaconis1988group}, for a precise statement and proof of this result, and for many interesting applications, e.g.~to card shuffling.}

One useful way to view this uniformization is through the lens of the information theory. 
As $p_t$ uniformizes with increasing $t$, it \textit{gains entropy} and \textit{loses information}.
In the $t\rightarrow\infty$ limit, $p_t$ becomes the maximum entropy distribution allowed by the symmetry requirement that its support must remain $H^\perp$. 
In this limit $p_t$ has lost as much information as is allowed by symmetry.
The only information remaining, which is protected by the symmetry, is the information about $H^\perp$ (and hence $H$).
This information is directly encoded in the support of $p_t$, which is the single degree of freedom defining a uniform distribution.
We can determine this degree of freedom by sampling $p_t$ for large enough $t$, as is done in the hidden cut algorithm.


\section{Heuristics to find approximate cuts}\label{sec:heuristics}

Armed with a deeper understanding, we will now turn to the task of deriving two heuristics. Both require running the hidden cut circuit with few input copies; enough to accentuate the features of the subsystem purity function $P(\bm s)$, but not eliminate its structure. Both heuristics are inspired by the idea of constructing an ``approximate nullspace'' when postprocessing the samples from the quantum computer. However, they differ in the final product: The first heuristic essentially stops the nullspace construction from Fourier samples before the trivial nullspace $\{00..0, 11..1\}$ is reached. The second one computes the ``average orthogonality'' of all samples to a possible solution $\bm s$, thereby constructing a classically efficient estimator of the reward function $P^t(\bm s)$, which we show to be statistically identical to estimating $P^t(\bm s)$ for each evaluation with a swap test on a quantum computer.

\subsection{Early stopping}\label{sec:heuristic1} 

\begin{figure}
    \centering
    \includegraphics[width=0.9\linewidth]{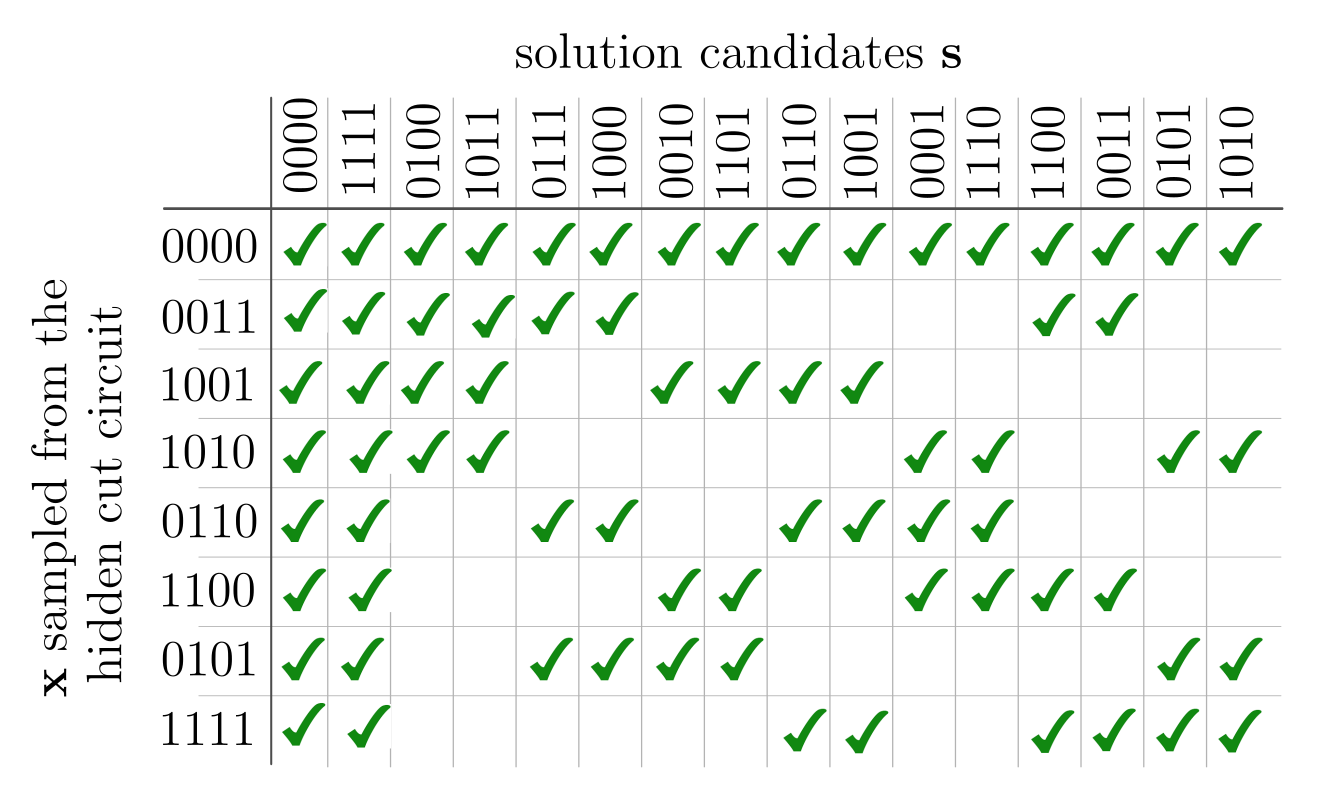}
    \caption{Postprocessing rules for possible samples drawn from the output distribution $p_t(\bm x)$ from Figure~\ref{fig:psi_ent}: A green box indicates that the sample $\bm x$ is orthogonal to a string $\bm s$ representing a subsystem, and hence admitted as a solution of the hidden cut algorithm. The columns are in descending order of $p_t(\bm x)$, while the rows are in descending order by $P(\bm s)$. It is clear from the plot that that samples with a high measurement probability admit subsystems with a high purity.}
    \label{fig:ticks}
\end{figure}

Recall that during postprocessing in the hidden cut algorithm, every nontrivial sample from $p_t(\bm x)$ is used to eliminate half of all possible subsystems $\bm s$, namely  those that fulfill $\bm x\cdot \bm s = 1$. For example, Figure~\ref{fig:ticks} shows the solution candidates admitted by any of the possible samples from the hidden cut circuit in the example shown in Figure~\ref{fig:psi_ent}.

Rewriting Eq.~(\ref{eq:main_result}) as
\begin{align} 
p_t(\bm x) &= \frac{1}{2^{n}} \left( \sum\limits_{\bm s|\bm x\cdot \bm s = 0}  P^t(\bm s) - \sum\limits_{\bm s|\bm x\cdot \bm s = 1}  P^t(\bm s) \right)
\end{align}
reveals that the probability of sampling $\bm x$ from the hidden cut circuit with $2t$ input copies is the average of $P^t(\bm s)$ over all $\bm s$ that $\bm x$ \textit{admits} as solutions, minus the average of $P^t(\bm s)$ over all $\bm s$ that $\bm x$ \textit{eliminates} during postprocessing. As a consequence, the probability $p_t(\bm x)$ will be higher for those $\bm x$ suggesting better cuts. To illustrate this, we show the rows and columns of Figure~\ref{fig:ticks} in the descending order of $p_t(\bm x)$ and $P(\bm s)$, respectively, and see that samples admitting high-purity $\bm s$ are more likely.

\begin{figure*}
    \centering
    \includegraphics[width=0.9\linewidth]{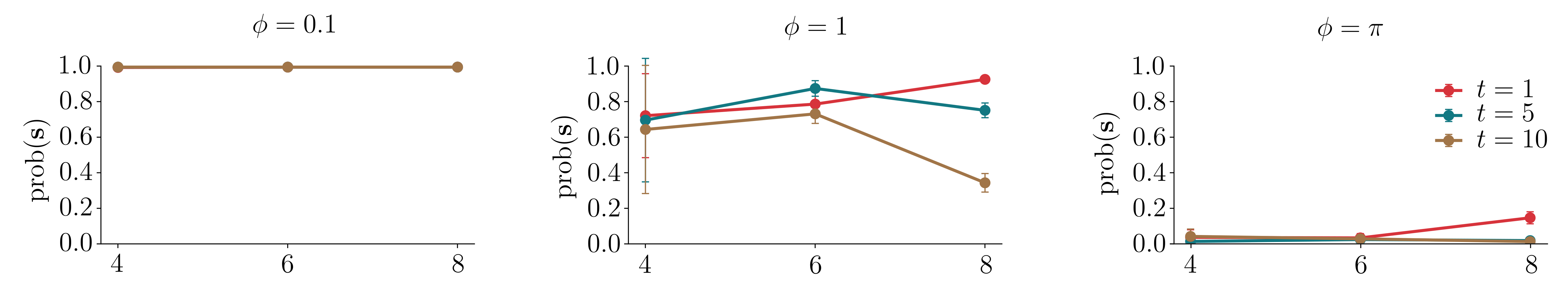}
    \caption{Probability of finding the planted cut. The state is composed of two $3$-qubit Haar random states that are entangled by a controlled $R_x(\phi)$ gate applied to the first qubit of the first register, and the last qubit of the last one, which creates a tunable cut at $\bm s = 000111, \bar{\bm s} = 111000$, while all other cuts are chance. We show mean and standard variance over $6$ different seeds to construct this state. The heuristic comfortably finds an approximate cut ($\phi=0.1$), and is even able to find a relatively weak cut ($\phi=1$) at small scales. Whether this is possible at larger scales critically depends on the entanglement \textit{within} the subsystems: if there are many good solutions, the heuristic will tend to find others than the best one. When the planted entanglement gets strong ($\phi=\pi$), the cut disappears and the heuristic finds other solutions. It is curious that the $t=1$ solution seems to improve with larger scales.} 
    \label{fig:prob_solution}
\end{figure*}

This observation suggests postprocessing the samples in an online fashion (i.e., one-by-one, possibly in order of their frequency) and stopping before the nullspace $\{\bm s| \bm x\cdot \bm s = 0 \;\; \forall \bm x \sim p_t(\bm x)  \}$ becomes trivial. Since the outcome is stochastic, this procedure should be repeated several times, each time yielding another solution candidate for the overall cut structure. The solution candidates are then aggregated into a final solution. As explained before, the number of copies $t$ serves as a hyperparameter that allows us to accentuate the purity function (similar to other quantum optimization strategies that apply a polynomial, such as \cite{jordan2025optimization}).

A concrete implementation of this strategy for a given number of copies $2t$ could look as follows:

\begin{enumerate}
    \item Run the hidden cut circuit several times to obtain $l$ samples $\bm x$. Make a list of unique samples, ordered by frequency. 
    \item Pick the most frequent $\bm x$ and eliminate all possible solution candidates $\bm s$ that are not orthogonal to $\bm x$. Discard this sample. 
    \item Repeat 2. and stop \textit{before} the solution candidates only contain the trivial all-zeros and all-ones solution, or when unique samples run out. \\
    \textit{Ex.: We get $\{0000, 0010, 1101, 1111\}$.}
    \item Extract the cut that this nullspace suggests.\\
    \textit{Ex.: The nullspace suggests that $\ket{\psi} \approx \ket{\psi_{013}}\ket{\psi_2}$.}
    \item Repeat Steps 1-3 $k$ times. \\
    \textit{Ex.: We might find the above solution $40\%$ of the time, and otherwise $\ket{\psi} \approx \ket{\psi_{012}}\ket{\psi_3}$.}
    \item Merge all cuts that were observed in more than $10\%$ of the runs. \\
    \textit{Ex.: The final solution is $\ket{\psi} \approx \ket{\psi_{01}} \ket{\psi_{3}}\ket{\psi_{2}}$}.
\end{enumerate}

Note that the last aggregation step could also result in hierarchical solutions, for example, that there is a stronger cut between qubits $013$ and $2$, than between qubits $012$ and $3$. 
Also note that the probability of drawing the uninformative all-zeros sample from the quantum algorithm,
\begin{align}
    p_t(00\dots 0) = \frac{1}{2^n}\sum_{\bm s} P^t(\bm s),
\end{align}
does not exponentially drown the signal with growing system sizes. Increasing the number of copies---and hence the sparsity of $P^t(\bm s)$---makes this value 
actually smaller.

Of course, the probability that the heuristic finds a certain cut structure depends on the details of the postprocessing strategy, and of the overall entanglement structure of the quantum state. We believe that it is difficult to derive a generic, but still meaningful analytic expression for the probability of a solution $\bm s$ being in the ``approximate nullspace'' after Steps 1-3 for all but simplified (and potentially misleading) cases. We therefore provide some small-scale numerical experiments here, but defer a thorough analysis of the scaling behaviour to application-specific follow-up work. We expect the heuristic to perform well at scale only when the cut structure is very prolific (i.e., the maxima in $P(\bm s)$ we aim to find are clearly higher than all other values, so that the polynomial suppression beats the growing scale). Variationally or otherwise amplifying this structure may broaden the cases in which entanglement can be detected. 

\begin{figure*}
    \centering
    \includegraphics[width=0.75\linewidth]{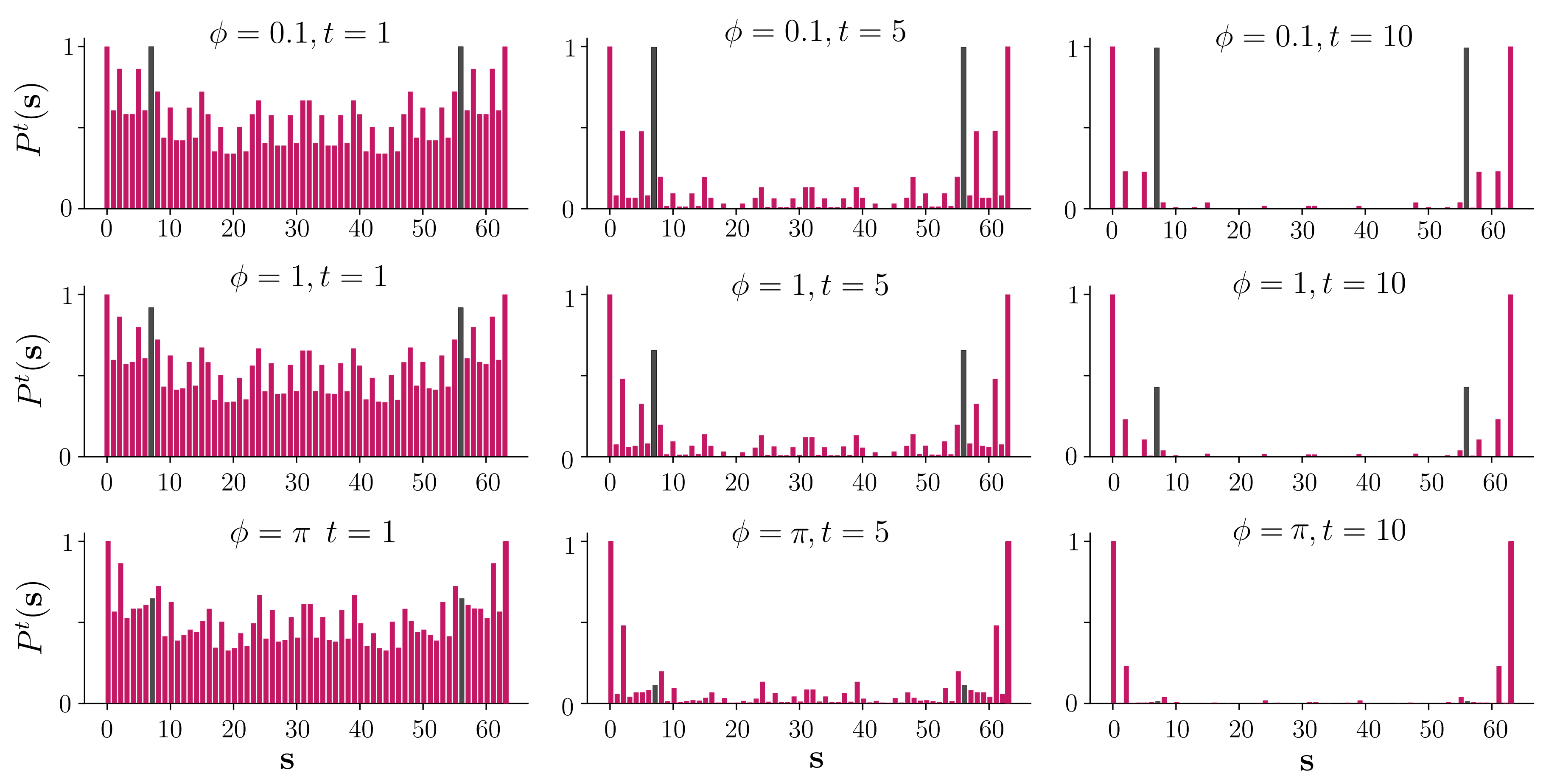}
    \caption{Subsystem purity function $P^t(\bm s)$ for different cut strengths $\phi$ and number of input state pair copies $t$ for one seed used in Figure~\ref{fig:prob_solution} at $n=6$ qubits. The subsystem bitstring pair corresponding to the planted cut is marked in gray, and its purity gets weaker compared to others with growing $\phi$. Using more input state pairs $t$ suppresses smaller purities and makes strong cuts stand out more.}
    \label{fig:purities}
\end{figure*}

\begin{figure*}
    \centering
    \includegraphics[width=0.75\linewidth]{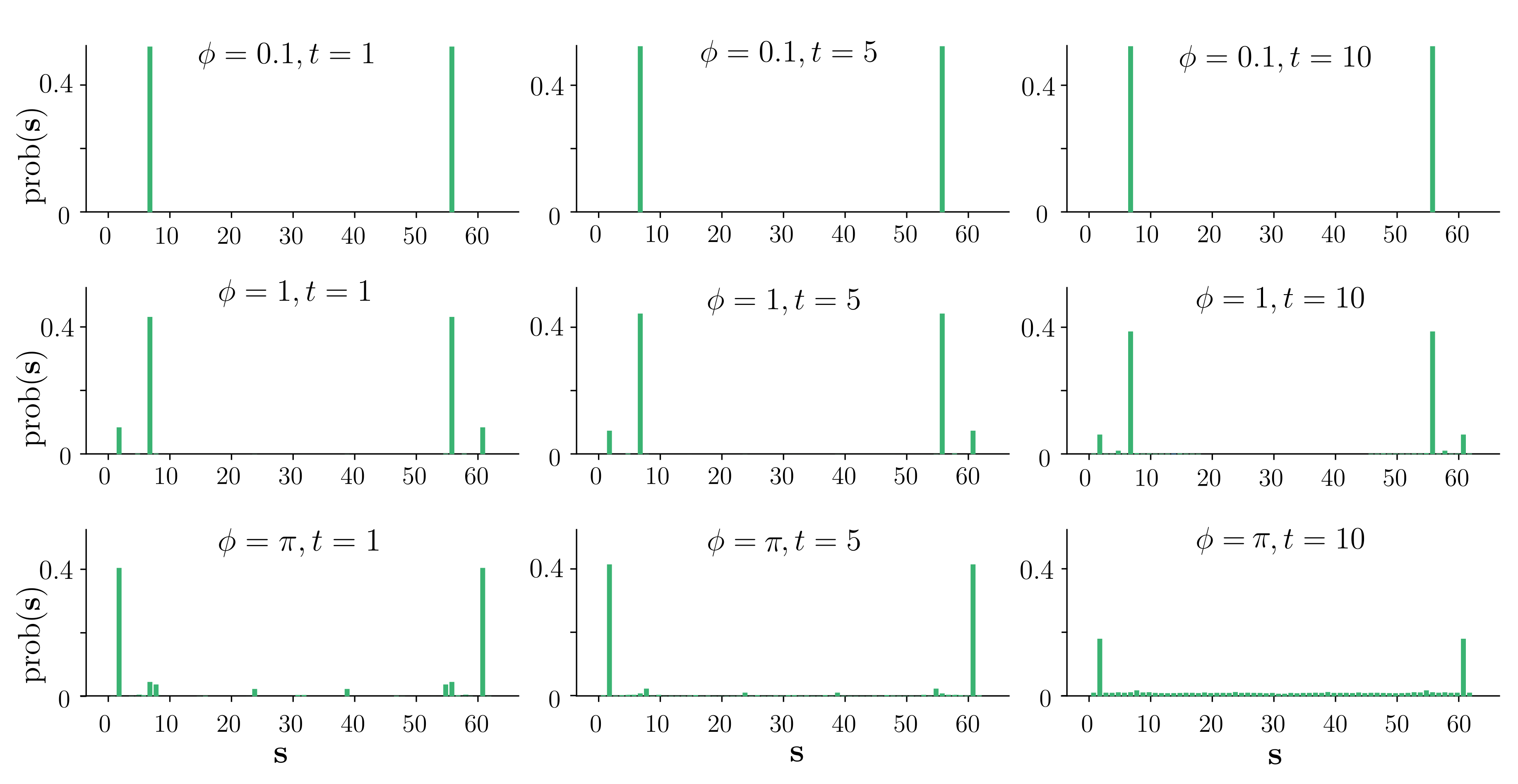}
    \caption{Probability $\text{prob}(\bm s)$ to find a subsystem $\bm s$ in the nullspace of Step 3 of the heuristic procedure, for the same settings as in Figure~\ref{fig:purities}. As the planted cut gets weaker ($\phi=1$) a secondary solution has a chance of getting found. When the cut disappears ($\phi=\pi$), that solution becomes the largest nontrivial maximum and also the most likely solution. The probabilities were empirically estimated to precision $\epsilon = 0.05$ using $2^6/\epsilon^2$ repetitions of Steps 1-3 of the heuristic.  }
    \label{fig:prob_s}
\end{figure*}

For small scales, in any case, the heuristic seems to work well. Figure~\ref{fig:prob_solution} shows the scaling behaviour up to $8$ qubits for a planted approximate cut where we vary the cut strength and number of copies.\footnote{Experiments are necessarily limited since each datapoint in the plots requires the evaluation of $2^n$ probabilities to precision $\epsilon$, and hence uses $m=\frac{2^n}{\epsilon^2}$ iterations of the hidden cut algorithm.} At these small scales it is obvious that the algorithm comfortably finds the strongest non-trivial cut, even if we entangle two separate subsystems with a controlled $R_x(\phi)$ rotation with angle $\phi=1$. As expected, for $\phi=\pi$---which allows to fully entangle the registers---the cut disappears.

To understand the behaviour better, Figures~\ref{fig:purities} and ~\ref{fig:prob_s} show the full subsystem purity (or reward) function, as well as the probability $\text{prob}(\bm s)$ of finding $\bm s$ in the nullspace when stopping the elimination procedure early (i.e., after Step 3 above), for $n=6$ qubits. For a cut strength of $\phi=1$, the first non-trivial maximum in the purity function is very close to the second maximum, and the algorithm actually starts to find the next most prominent cut at $\bm s = 000010, \bar{\bm s} = 111101$ with high probability---arguably something expected from the algorithm when this is a viable solution. As expected, once there is no prolifically strong cut any more ($\phi=\pi$), the heuristic cannot give us much information.

\subsection{Constructing an estimator of the subsystem purities for classical optimization}\label{sec:postprocessing}

As an alternative strategy, we will show how to use the samples of the hidden cut algorithm to construct an estimator $\hPt$ of $P^t$. This can be interesting in situations where we want to probe the properties of the reward or subsystem purity function on a classical computer---either to query $\hPt(\bm s)$ at desired values $\bm s$, or to use classical optimizers on $\hPt$ without having to make measurements at each step.

The key is the following lemma:

\begin{lemma}
    Let $p(\bm x)$ be a probability distribution over $\{0,1\}^n$
    and let $f(\bm s)\equiv \sum_{\bm x} (-1)^{\bm x\cdot \bm s} p(\bm x)$ be its Fourier transform.
    Then
    \begin{align}
        \mathbb E_{\bm x\sim p}\left[\bm x\cdot \bm s\right]
        =
        \frac{1}{2}\left(1 - f(\bm s)\right),
    \end{align}
    where we recall that $\cdot$ denotes the dot product modulo 2.
\end{lemma}

\begin{proof}
    With $(-1)^{\bm x\cdot \bm s} = 1 - 2(\bm x\cdot \bm s)$ we get:
    \begin{align}
        \mathbb E_{\bm x\sim p}\left[\bm x\cdot \bm s\right]
        &\equiv
        \sum_{\bm x} p(\bm x)(\bm x \cdot \bm s)\\
        &=
        \sum_{\bm x} p(\bm x)\frac{1}{2}\left[1 - (-1)^{\bm x\cdot \bm s}\right]\\
        &=
        \frac{1}{2}\left(1 - f(\bm s)\right).
    \end{align}
\end{proof}

Applying the lemma to the distribution $p_t(\bm x)$ with the Fourier transform $P^t(\bm s)$ gives
\begin{align}\label{eq:mean_x_dot_s}
    \mathbb E_{\bm x\sim p_t}\left[\bm x\cdot \bm s\right]
    =
    \frac{1}{2}\left(1 - P^t(\bm s)\right),
\end{align}
or, rearranging,
\begin{align}
    P^t(\bm s) = 1 - 2\mu_{\bm x\cdot\bm s},
\end{align}
where $\mu_{\bm x\cdot \bm s}\equiv \mathbb E_{\bm x\sim p_t}\left[\bm x\cdot \bm s\right]$ is the mean of the random variable $\bm x\cdot \bm s$.
This is a random variable because $\bm x$ is a random variable (with distribution $p_t(\bm x)$).
Moreover, because the dot-product is modulo 2, $\bm x\cdot \bm s$ can only take on the values 0 and 1:
it is a Bernoulli random variable.
Therefore, if we can estimate the mean of this Bernoulli random variable, we can use the above equation to estimate $P^t(\bm s)$.

Now, $m$ measurements from the $t$-copy hidden cut circuit give us $m$ $n$-bit samples $\{\bm x_1,\cdots\bm x_m\}$.
 Arrange these samples in a binary \textit{measurement matrix} $M$ of shape $m\times n$, 
in which the $k$-th row of $M$ is equal to $\bm x_k$.
Then we can estimate $\mu_{\bm x\cdot \bm s}$ as the sample mean $\hatmu$,
\begin{align}
    \hatmu\equiv \frac{1}{m} \sum_{k=1}^m\bm x_k\cdot\bm s
    =\frac{|M\bm s|}{m},
\end{align}
where $|M\bm s|$ is the Hamming weight (number of 1s) of the $m$-bit vector $M\bm s$.
This gives us an estimator for $P^t(\bm s)$, namely,
\begin{align}
    \hPt(\bm s)=1 - 2\hatmu,
\end{align}
or, explicitly in terms of $M$,
\begin{align}\label{eq:hat_P^t}
    \hPt(\bm s)=1 - \frac{2|M\bm s|}{m}.
\end{align}
The dependence of the right-hand side on $t$ is implicit through the distribution $p_t$ from which the samples are drawn.

\begin{figure}
    \centering
    \includegraphics[width=0.7\linewidth]{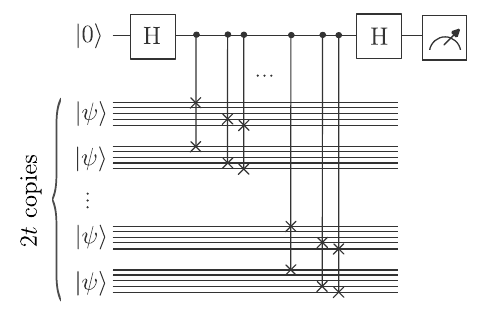}
    \caption{Swap test circuit to estimate $\widehat{P^t(\bm s)}$ for a fixed $\bm s$. Using samples from the hidden cut circuit, we derive an estimator of the same quality, but without the need to run a different circuit for each new input $\bm s$.}
    \label{fig:swap_test_circuit_t}
\end{figure}

Note that for any fixed value $\bm s$, we can estimate $P^t(\bm s)$ by running a Swap test circuit that swaps the qubits corresponding to 1's in $\bm s$, which is shown in Figure~\ref{fig:swap_test_circuit_t}. We can write such an estimator as $\widehat{P^t(\bm s)}$, and show in Appendix~\ref{app:statistical_analysis} that they are statistically identical. However, as shown in Figure~\ref{fig:tradeoff}, $\hPt(\bm s)$ requires only a constant number of shots to estimate $P^t$ for any values of $\bm s$.

The constructor of an estimator reveals an important link between the problem of finding the strongest approximate hidden cut and \textit{coding theory} (see, e.g.~\cite{macwilliams1977theory}). In this language, the space of all vectors $\bm s$ is the \textit{ambient space}, 
$M$ is the \textit{parity check matrix}, 
$M\bm s$ is the \textit{syndrome} of $\bm s$, 
$\ker(M)$ is the \textit{codespace}, 
and the vectors in $\ker(M)$ are the \textit{codewords}. 
Each row of the parity check matrix represents a parity condition that any valid codeword must satisfy. Finding approximate hidden cuts is precisely the \textit{minimum syndrome weight problem} in coding theory, 
which asks for the non-trivial ambient vectors $\bm s^*$ and $\bar{\bm s}^*$ that are as close to the codespace as possible. This perspective also illuminates the computational complexity of finding the maximum nontrivial cuts $\bm s^*$ and $\bar{\bm s}^*$.
There are two cases, with very different complexities. 
In the case where $M$ has a non-trivial kernel, 
$\bm s^*$ and $\bar{\bm s}^*$ are the non-trivial elements $\ker(M)$, 
and they can be efficiently found by solving $M\bm s=0$ via Gaussian elimination.
But in the general case, where there are no non-trivial solutions to $M\bm s=0$ and we do not assume any other particular structure on $M$, 
the minimum syndrome weight problem is NP-hard.

\begin{figure}[t]
    \centering
    \includegraphics[width=\linewidth]{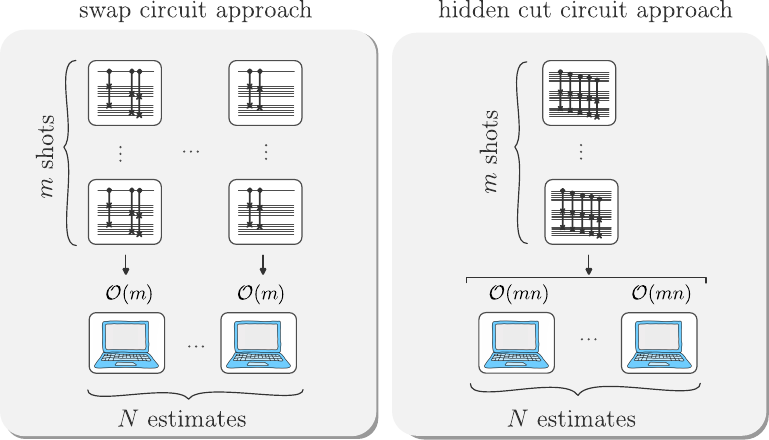}
    \caption{Schematic of the resource trade-off between constructing an estimator for the subsystem purity function $P^t(\bm s)$ using $m$ samples from the hidden cut circuit, or from the Swap test circuit shown in Figure~\ref{fig:swap_test_circuit_t}. The hidden cut circuit can be seen as a parallelized version of the swap circuit approach, where measurement samples contain information on the purity of all subsystems $\bm s$. This allows $m$ shots to be used to make any number of $N$ estimates, which can be useful when optimizing over the purity function.}
    \label{fig:tradeoff}
\end{figure}

We finish with one comment that might be interesting for applications. Writing out the matrix multiplication and Hamming norm explicitly, this is
\begin{align}
    \hPt(\bm s) = 
    1 - \frac{2}{m}
    \sum_{i=1}^m\left[\left(\sum_{j=1}^n M_{ij} s_j\right)\text{ mod }2\right].
\end{align}
Now define
\begin{align}
    W^{(1)} &\equiv \text{$m\times n$ matrix with entries $M_{ij}$,} \\
    W^{(2)} &\equiv \text{$1\times m$ matrix with constant entries $-2/m$,} \\
    b^{(1)} &\equiv \text{$m$-dim vector of all zeros,} \\
    b^{(2)} &\equiv 1,\\
    f(x) &= x\text{ mod }2.
\end{align}
Then
\begin{align}
    \hP^t(\bm s)
    = 
    b^{(2)} + W^{(2)}f\left(b^{(1)} + W^{(1)}\bm s\right).
\end{align}
Thus $\hPt(\bm s)$ can be thought of as a special two-layer neural network with input dimension $n$, hidden dimension $m$, output dimension $1$, and non-linear activation function $f$,
where the first layer has a binary weight matrix and no bias term, 
while the second layer has a constant weight matrix and a bias term of 1. The weights were not learned, but instead obtained in a theoretically principled way as measurement samples from the hidden cut quantum circuit. In contrast, there have been many previous works that use machine learning to study the correlation structure of many-body quantum systems, e.g.~\cite{torlai2018neural, carleo2017solving, schindler2017probing, berkovits2018extracting, chen2021detecting, lin2023quantifying, lee2024estimating, schmidt2025transfer, saleh2025predicting}. The estimator $\hP(\bm s)$ could be useful for machine learning approaches such as these, for example as an initialisation strategy for the weight matrices.

\section{Conclusion}

We want to conclude with a few comments on how the insights from this paper might be useful for quantum applications. 

Understanding the entanglement structure of many-body quantum systems is a central problem in many domains of physics \cite{amico2008}. In condensed matter physics, for example, entanglement provides a way to characterize and classify phases of quantum matter that are indistinguishable using traditional symmetry-based methods \cite{kitaev2006, levin2006, hui2008}; in quantum field theory, entanglement from the vacuum can be harvested by local detectors coupling to the field \cite{summers1985vacuum, salton2015acceleration, pozas-kerstjens2016entanglement, perche2024fully},
and can even be used as a resource to teleport energy between different regions of spacetime \cite{hotta2008a, hotta2008b} and in quantum gravity, subregion entanglement of the boundary of spacetime is intimately related to the geometry of the bulk of the spacetime \cite{ryu2006holographic, hubeny2007covariant, VanRaamsdonk2010building},
a fact that is central to proposed resolutions \cite{penington2020entanglement, almheiri2020replica} of the black hole information paradox \cite{hawking1974black}.  Quantum chemistry might use quantum states where qubits model electrons. Some of the above examples may allow for efficient simulation on quantum computers in future, and the hidden cut heuristics developed here might help us to understand and shape the entanglement structure of the emulated system.

While such physics-based use cases are fairly obvious, we might also encounter situations where qubits model classical variables, and detecting entanglement could help us understand their statistical independence. For example, consider a generative quantum machine learning model \cite{liu2018differentiable, amin2018quantum, recio2025train} $\ket{\psi_{\theta}}$ with trainable parameters $\theta$, that produces new data by sampling from the measurement distribution $p(\bm x) = |\langle \bm x |\psi_{\theta}\rangle|^2$. If the state is separable across a cut, $\ket{\psi_{\theta}} = \ket{\psi_{\bm s}} \ket{\psi_{\bar{\bm s}}}$, the generative model will factorize in the same manner, $p(\bm x) = p_{\bm s}(\bm x) p_{\bar{\bm s}}(\bm x)$. The signal extracted from the hidden cut algorithm could for example be used to encourage a disentangled (or ``statistically simple'') model during training. Another example are (hyper-)graphs encoded into quantum states, for example via graph states \cite{hein2006entanglement}. Here, the analysis of separable qubit registers could uncover certain types of community structure.

Lastly, we want to remark that the way that quantum algorithms for hidden subgroup problems work is rather special: instead of encoding the solution into a sample or an expectation value, they are genuinely hybrid algorithms where a few samples from the quantum computer contain the solution. This opens up a different playing field, where the outcome of a quantum algorithm is not a generation task (as, for example, in quantum optimization \cite{jordan2025optimization, farhi2014quantum}), but still only requires a few shots. Furthermore, quantum algorithms for hidden subgroup problems exploit group structure, which quantum computers are natively suited to (i.e., \cite{aaronson2022much}). Despite these qualities, interest in algorithms for algebraic problems has drastically faded since the 2000's, as perfect input structures are rarely encountered in the ``wild''. Our results suggest that a fruitful path to quantum applications could be to go back to these origins of quantum computing, but this time from a heuristic angle: one where perfect structure is replaced by approximate structure, the requirement to find the correct solution is replaced by trying to find good enough solutions, and provable worst-case guarantees by theoretically well-motivated heuristics.

\acknowledgments

We gratefully acknowledge David Wakeham as a collaborator during the early stages of this project. The Fidelity publishing approval number for this paper is 1253788.1.0.

\bibliography{references}

\appendix

\section{General abelian StateHSP}\label{app:abelian_StateHSP}

Recall that the hidden cut circuit is a specific instance of the general StateHSP circuit, 
with a particular choice of group $G$, namely $G=\mathbb Z_2^n$, 
and a particular choice of unitary representation of $G$, namely $U(\bm s) = \text{SWAP}_{\bm s}^{\otimes t}$.
Here we will generalize many of our key results obtained for the hidden cut circuit to the case of StateHSP circuits
with arbitrary abelian group $G$ and arbitrary unitary representation of $G$.
To more closely parallel the derivations for the hidden cut circuit, 
let us denote the unitary representation as $U(g)^{\otimes t}$, where $t$ is an arbitrary positive integer.
Note that there is no loss of generality with this class of representations, 
since we can always set $t=1$ and $U(g)$ to any desired unitary representation.

Let us now trace through the general abelian StateHSP circuit in Figure~\ref{fig:StateHSP_circuit}, just as we did for the hidden cut circuit.
The input state to the circuit is $\ket{g}\ket\Psi$, 
where $\ket{g}$ encodes the group elements and $\ket\Psi$ is the state whose hidden symmetry the StateHSP algorithm seeks to find.
Before the final Fourier transform the full state is
\begin{align}
    \frac{1}{\sqrt{|G|}}\sum_{g\in G} \ket{g}
    \left(U(g)\ket\Psi\right)^t.
\end{align}
After the final Fourier transform the state becomes
\begin{align}
    \frac{1}{\sqrt{|G|}}\sum_{g\in G}
    \sum_{k\in G}
    F_{kg}\ket{k}
    \left(U(g)\ket\Psi\right)^t,
\end{align}
where $F_{kg}$ are the components of the unitary matrix $F\in \mathbb C^{|G|\times|G|}$ that defines the Fourier transform over the group $G$. 
The fundamental theorem of finite abelian groups says that any abelian $G$ is isomorphic to $\mathbb Z_{N_1}\times\dots \times \mathbb Z_{N_m}$, where the $N_i$ are prime~\cite{dummit2004abstract}.
Under this isomorphism $g$ and $k$ can be indexed as $k=(k_1,\dots,k_m)$ and $g=(g_1,\dots,g_m)$ where $k_i, g_i\in \{0,\dots, N_i\}$, 
and $F_{kg}$ is given by
\begin{align}
    F_{kg} = \frac{1}{\sqrt{|G|}}
    e^{-2\pi ik_1g_1/{N_1}}
    \cdots
    e^{-2\pi ik_mg_m/{N_m}}.
\end{align}

Measuring the first (group) register results in outcome $k$ with probability
\begin{align}
    p_t(k) = \frac{1}{|G|}\sum_{g'} \sum_{g''} F_{kg'} F^*_{kg''}
    \left(\bra\Psi U^\dagger(g'')U(g')\ket\Psi\right)^t.
\end{align}
The definition of the matrix elements $F_{kg}$ implies $F_{kg'}F^*_{kg''} = F_{k(g'g''^{-1})}/\sqrt{|G|}$,
and the fact that $U$ is a unitary representation of an abelian group implies $U^\dagger(g'')U(g') = U(g'g''^{-1})$.
Therefore
\begin{align}
    p_t(k) = \frac{1}{|G|}\sum_{g'} \sum_{g''} \frac{1}{\sqrt{|G|}}F_{k(g'g''^{-1})} 
    \bra\Psi U(g'g''^{-1})\ket\Psi,
\end{align}
As a check notice that $\sum_k p_t(k) = 1$, as expected.
Letting $g\equiv g'g''^{-1}$, the double sum simplifies to a single sum:
\begin{align}\label{eq:p_t_from_Psi_general_G}
    p_t(k) &= \frac{1}{\sqrt{|G|}} \sum_{g} F_{kg}
    \left(\bra\Psi U(g)\ket\Psi\right)^t \\
    & \equiv
    \mathcal F^{-1}[\left(\bra\Psi U(\cdot) \ket\Psi\right)^t](k),
\end{align}
where we define the Fourier and inverse Fourier transforms over $G$ as
\begin{align}
    \mathcal F[f](g) & \equiv \sqrt{|G|}\sum_k F^*_{gk} f(k),\\
    \mathcal F^{-1}[\hat f](k) & \equiv \frac{1}{\sqrt{|G|}}\sum_g F_{kg} \hat f(g).
\end{align}
Letting $t=1$ gives the useful expression
\begin{align}\label{eq:p_1_FT_general_G}
    p_1 = \mathcal F^{-1}[\bra\Psi U(\cdot)\ket\Psi],
\end{align}
which is in fact a fully general expression since $U$ is an arbitrary representation.

We can relate $p_t$ to $p_1$ as we did in the hidden cut circuit.
First we invert \eqref{eq:p_1_FT_general_G} to get
\begin{align}
    \bra\Psi U(g) \ket\Psi =
    \mathcal F[p_1(\cdot)](g).
\end{align}
Substituting this into~\eqref{eq:p_t_from_Psi_general_G} gives $p_t$ in terms of $p_1$,
\begin{align}
    p_t = \mathcal F^{-1}
    \left[\mathcal F[p_1]^t\right].
\end{align}
Defining the convolution operator as
\begin{align}
    (f\ast h)(k)\equiv \sum_{k'}f(k')h(kk'^{-1}),
\end{align}
the convolution theorem thus gives
\begin{align}
    p_t = p_1^{\ast t}
    \equiv 
    \underbrace{p_1 \ast \cdots \ast p_1}_{\text{$t$ times}},
\end{align}
exactly the same as in the special case of the hidden cut circuit.

\section{Heuristic 2: Statistical analysis of the purity estimator}\label{app:statistical_analysis}

We provide some more details to compare $\hPt(s)$, the estimator of the subsystem purity function constructed from hidden cut circuit samples, and $\widehat{P^t(s)}$ the estimator of a value of this function constructed from swap test circuit samples.

\subsection{Equivalence of $\hPt(s)$ and $\widehat{P^t(s)}$}

In this section we compare the hidden-cut-based estimator $\hPt(\bm s)$ to the swap-test-based estimator $\hPts$.
We will find that there is no difference between them: they are statistically identical to one another.

To see this, recall that
\begin{align}
    \hPts &= 1 - 2\hat\mu_x,\\
    \hPt(\bm s) &= 1 - 2\hatmu, \label{eq:hPt(s)=1-2_hatmu}
\end{align}
where $\hat\mu_x$ and $\hatmu$ are sample means of the Bernoulli random variables $x\sim p_t(x)$, and $\bm x\cdot \bm s$ with $\bm x\sim p_t(\bm x)$.
Here $p_t(x)$ is the single-bit output distribution of the $t$-copy swap test circuit, 
and $p_t(\bm x)$ is the $n$-bit output distribution of the $t$-copy hidden cut circuit.

Now, the distribution of a Bernoulli variable $X$ is fully specified by its mean, 
which is equal to the probability that $X=1$.
The mean value of $x$ is
\begin{align}
    \mu_x=\frac{1}{2}\left(1 - P^t(\bm s)\right),
\end{align}
and from Eq.~\eqref{eq:mean_x_dot_s} the mean value of $\bm x\cdot \bm s$ is
\begin{align}
    \mu_{\bm x\cdot \bm s}=\frac{1}{2}\left(1 - P^t(\bm s)\right).
\end{align}
Thus the Bernoulli variables $x$ and $\bm x\cdot \bm s$ have the same mean, and hence they follow the same distribution: they are are statistically indistinguishable random variables.
This implies that their $m$-sample means $\hat\mu_x$ and $\hatmu$ are indistinguishable,
and finally that the estimators $\hPts=1 - 2\hatmu$ and $\hPt(\bm s) = 1 - 2\hat\mu_x$ are indistinguishable.
Therefore there is no statistical difference in estimating $P^t(\bm s)$ as $\hPt(\bm s)$ (hidden cut approach) 
versus $\hPts$ (swap test approach).

\subsection{Distribution of $\hPt(s)$ and $\widehat{P^t(s)}$}

We have just shown that the estimators $\hPts$ and $\hPt(\bm s)$ follow the same distribution. What is this distribution?

Let us work with the estimator $\hPt(\bm s)$, since the distribution for $\hPts$ is the same.
Recall that $\hPt(\bm s)=1-2\hatmu$, 
where $\hatmu$ is the $m$-sample mean of the Bernoulli random variable $\bm x\cdot \bm s$.
Because $m\hatmu$ is a sum of $m$ i.i.d.~samples of a Bernoulli variable, by definition it follows a Binomial distribution,
\begin{align}
    m\hatmu &\sim \text{Binomial}\left(m, \mu_{\bm x\cdot\bm s}\right),\\
\end{align}
where $\mu_{\bm x\cdot\bm s} = \frac{1}{2}\left(1 - P^t(\bm s)\right)$ is the true mean of $\bm x\cdot \bm s$.
Therefore $\hPt(\bm s)$ follows an \textit{affinely transformed Binomial distribution}, 
that is a Binomial distribution that is scaled and shifted by a constant. 
More precisely, the random variable $\hPt(\bm s)$ takes on the values $1-\frac{2k}{m}$, where $k\in \{0,1,\cdots,m\}$,
with probabilities
\begin{align}
    \text{Prob}\left(\hPt(\bm s)=1-\frac{2k}{m}\right)
    &=
    \binom{m}{k}
    \mu_{\bm x\cdot\bm s}^k (1-\mu_{\bm x\cdot \bm s})^{m-k},
\end{align}
where $k\in\{0,1,\dots,m\}$.
To get samples of $\hPt(\bm s)$ we:
\begin{enumerate}
    \item Sample $m\hatmu\sim\text{Binomial}\left(m, \mu_{\bm x\cdot\bm s}\right)$,
    \item Multiply by $2/m$, and
    \item Subtract the result from 1.
\end{enumerate}
Thus we have fully specified the distribution of $\hPt(\bm s)$, and therefore also $\hPts$.
\subsection{Bias, variance, and the signal-to-noise ratio}

To get some intuition about the distribution of $\hPt(\bm s)$, we can compute its mean and variance.
The mean is
\begin{align}
    \mathbb E\left[\hPt(\bm s)\right] 
    = 1 - 2\mathbb E[\hatmu]
    = 1 - 2\mu_{\bm x\cdot\bm s} 
    = P^t(\bm s),
\end{align}
which shows that $\hPt(\bm s)$ is an unbiased estimator of $P^t(\bm s)$,
and the variance is
\begin{align}\label{eq:var_hat_P^t}
    \var\left(\hPt(\bm s)\right)
    &= 4\var\left(\hatmu\right)\\
    &= \frac{4\mu_{\bm x\cdot\bm s}\left(1 - \mu_{\bm x\cdot\bm s}\right)}{m^2}
    &= \frac{1}{m^2}\left(1 - P^{2t}(\bm s)\right).
\end{align}

Now consider the \textit{signal-to-noise} ratio
\begin{align}
    \text{SNR}\left(\hPt(\bm s)\right)
    \equiv
    \frac{\mathbb E[\hPt(\bm s)]}{\sqrt{\var(\hPt(\bm s))}}
    =
    m \frac{P^t(s)}{\sqrt{1 - P^{2t}(s)}}.
\end{align}
In general, an estimator is good when its SNR is much bigger than 1, and poor when its SNR is much less than 1.
We see that, as is often the case, increasing the number of samples $m$ linearly increases the SNR.
Hence we can always get a good estimator by taking a large enough $m$.
We also note that $\text{SNR}\rightarrow0$ when $P(s)\rightarrow 0$, and
$\text{SNR}\rightarrow\infty$ when $P(s)\rightarrow 1$,
so we have a good SNR for high purity subsystems and a bad SNR for low purity subsystems.
This is potentially concerning, because by Page's theorem the purity of a typical subsystem of size $k$ of a Haar random $n$-qubit state is $\mathcal O(2^{-k})$ \cite{page1993average}.
Thus to precisely estimate the purity of such a subsystem, we would need $m=\mathcal O(2^{k})$.
However, many states of interest are not Haar random, 
and many subsystems of interest have purities that are not exponentially small.
In such cases we would be able to obtain precise purity estimates with more reasonable values of $m$.

\end{document}